\documentclass{siamltex}
\usepackage[utf8]{inputenc}
\usepackage[english]{babel}
\usepackage[T1]{fontenc}
\usepackage{lmodern}
\usepackage{framed}
\usepackage{fancybox}
\usepackage{ulem}
\usepackage{csquotes}

\usepackage{calligra}                                           
\usepackage{mathrsfs} 
\usepackage{ dsfont }
\usepackage{hyperref}
\usepackage{xcolor}
\hypersetup{
    colorlinks,
    linkcolor={red!50!black},
    citecolor={blue!50!black},
    urlcolor={blue!80!black}
}

\usepackage{amsfonts}
\usepackage{amsmath}
\usepackage{amssymb} 
\usepackage{graphicx}
\usepackage{mathtools}
\usepackage[noabbrev]{cleveref}

\def\be#1#2\ee{\begin{equation}\label{eq:#1}#2\end{equation}}
\def\req#1{{\rm(\ref{eq:#1})}}

\newcommand*\diff{\mathop{}\!\mathrm{d}}

\newcommand*\dx{\diff x}
\newcommand*\dy{\diff y}

\newcommand*{\R}{\mathbb{R}}
\newcommand*{\N}{\mathbb{N}}
\newcommand*{\U}{\mathscr{U}}
\renewcommand{\P}{\mathscr{G}}
\newcommand{\gp}{\Omega}
\newcommand{\xx}{\mathbf{x}}
\newcommand{\yy}{\mathbf{y}}

\newcommand{\thsp}{\hspace*{0.1ex}}

\newtheorem{theoremold}{Theorem}

\title{A note on the uniqueness result\\[0.4ex]
       for the inverse Henderson problem\thanks{The 
    research leading to this work has been funded by the Deutsche 
    Forschungsgemeinschaft (DFG, German Research Foundation) -- 
    Projektnummer 233630050 -- TRR~146.}}
\author{Fabio Frommer\thanks{Institut f\"ur Mathematik, Johannes
    Gutenberg-Universit\"at Mainz, 55099 Mainz, Germany
    ({\tt fafromme@uni-mainz.de})} \and
    Martin Hanke\thanks{Institut f\"ur Mathematik, Johannes
    Gutenberg-Universit\"at Mainz, 55099 Mainz, Germany
    ({\tt hanke@math.uni-mainz.de})} \and 
        Sabine Jansen\thanks{Institut f{\"u}r Mathematik, 
        Ludwig-Maximilians Universit{\"a}t M{\"u}nchen, Theresienstr. 39, 
        80333 M{\"u}nchen, Germany ({\tt jansen@math.lmu.de})}. 
	}

\begin{document}
\reversemarginpar

\sloppy
\maketitle

\begin{abstract}
The inverse Henderson problem of statistical mechanics concerns classical 
particles in continuous space which  interact according to a pair potential 
depending on the distance of the particles. Roughly stated, it asks for the 
interaction potential given the equilibrium pair correlation function of the 
system. In 1974 Henderson proved that this potential is uniquely determined 
in a canonical ensemble and he claimed the same result for the 
thermodynamical limit of the physical system. Here we provide a rigorous proof 
of a slightly more general version of the latter
statement using Georgii's version of the Gibbs variational principle. 
\end{abstract}


\begin{AMS}
  {\sc 82B21}
\end{AMS}


\section{Introduction}
\label{Sec:intro}
An important inverse problem in computational physics and
computational chemistry is the determination of the interacting forces 
in a system of particles
in continuous space, given structural information on the
spatial distribution of the particles. In the simplest incarnation of this 
problem it is assumed that the potential energy of the particle ensemble is 
determined by a pair potential which only depends on the distance of the
interacting particles. In an often cited paper Henderson~\cite{RHEN74}
has claimed that under given conditions of temperature and density
this pair potential is uniquely determined by the
so-called \emph{radial distribution function}, which -- suitably normalized --
assigns to each $r>0$ the expected number of particles on a sphere of 
radius $r$ around any given particle. Roughly speaking,
the radial distribution function is obtained from the pair correlation function
(called pair density function in the physical literature) associated 
with a canonical or grand canonical ensemble in a finite volume, 
when taking the limit of the volume to infinity, the so-called
\emph{thermodynamical limit}. To give credit to Henderson's contribution,
this inverse problem of statistical mechanics is sometimes called the
inverse Henderson problem.

Henderson's argument makes use of a technique suggested
by Hohenberg and Kohn~\cite{HK64}, Mermin~\cite{DMER65}, and others,
for studying a similar inverse problem for \emph{external potentials}.
The key idea is to apply a \emph{Gibbs variational principle},
which states that
in a system with given thermodynamic conditions the associated 
thermodynamic potential becomes minimal, 
if and only if the distribution of the particles
is given by the probability measure associated with this system.
The particular version of this principle used by Henderson is based on
the free energy functional in a canonical ensemble, where the
finite volume pair correlation function and not the radial distribution 
function is the relevant stochastic quantity. To extend the uniqueness result 
to the radial distribution function,
Henderson subsequently turns to the thermodynamical limit, ignoring the
possibility that the strict inequality of the variational principle
for finite volumes may turn
into an equality when the volume tends to infinity. Accordingly, there is
a gap in the argument provided in \cite{RHEN74} --
aside of the fact that no mention is made concerning the necessary requirements 
for the pair potential, e.g., its behavior for particle pairs with
diminishing distances.

In this note we specify a suitably rich class of pair potentials, for which
Henderson's approach can be turned into a rigorous proof by using a version
of the Gibbs variational principle due to
Georgii~\cite{HGEO95}. This class of
potentials does include hard core potentials and
the so-called Lennard-Jones type potentials, but is
a strict subclass  of all superstable potentials, 
cf., e.g., Ruelle~\cite{DRUE69} for this and further terminology. 
Strictly speaking, this means that our result does not answer the question
whether the radial distribution function associated with, say, the
classical Lennard-Jones potential can also occur for a much more exotic 
type of pair potential and the same values of temperature and density.

The thermodynamical limit may either be reached from a canonical or a
grand canonical ensemble. We therefore also state a variant of Henderson's
result which is more natural from the grand canonical perspective:
It will be shown below that the pair potential is uniquely determined when 
given the temperature, the chemical potential, and the infinite volume
pair correlation function; it is unknown whether in this second version
of Henderson's statement the pair correlation function can be replaced by the 
radial distribution function in the isotropic case.

We mention that for the corresponding inverse problem on the lattice
the uniqueness of the pair potential has already been settled by
Griffiths and Ruelle~\cite{GR71}; 
see also Caglioti, Kuna, Lebowitz, and Speer~\cite{CKLS06}.
 
The outline of this note is as follows:
In Section~\ref{Sec:preliminaries} we review the rigorous mathematical
setting of the thermodynamical limit of a grand canonical ensemble when
the system is translation invariant and its potential energy is given
by pairwise interactions only.
Then we formulate in Section~\ref{Sec:Gibbs}
the particular version of the Gibbs variational principle that is valid
in this setting. Section~\ref{Sec:Henderson} is devoted to the 
proof of the uniqueness results, 
and eventually we close with a few comments and open problems in Section \ref{Sec:problems}.

\section{The thermodynamical limit of the grand canonical ensemble}
\label{Sec:preliminaries}
We start from a grand canonical ensemble of pointlike classical particles
in a bounded box $\Lambda_\ell = [-\ell,\ell]^d$, with specified
inverse temperature $\beta>0$ and chemical potential $\mu\in\R$. 
We restrict our attention to the case that the interaction of the particles 
is given by a pair potential $u: \R^d \to \R\cup\{+\infty\} $, which is an even 
function, i.e., $u(x)=u(-x)$, satisfying the following two assumptions:
\begin{enumerate}
\item There exists $r_0>0$ and a decreasing function 
      $\varphi\colon (0,r_0]\to \R^+_0$ with
	\begin{align*}
		 \int_{0}^{r_0} r^{d-1}\varphi(r)\diff r = +\infty
	\end{align*}
	and 
\[
   u(x) \geq \varphi(|x|) \qquad \textrm{ for }\ |x| \leq r_0.
\]
\item There exists a decreasing function 
      $\psi\colon [r_0,\infty)\to \R^+_0$ with
	\begin{align*}
		 \int_{r_0}^\infty r^{d-1}\psi(r)\diff r <\infty
	\end{align*}
	and 
\be{majorant}
   |u(x)| \leq \psi(|x|) \qquad \textrm{ for }\ |x| \geq r_0.
\ee
\end{enumerate}
For this class $\mathscr{U}$ of potentials the associated configurational Hamiltonian
of $m\in\mathbb{N}_0$ particles at positions $x_i\in\R^d$, $i=1,\dots,m$, given by
\[
	H_u(\mathbf{x}_m) = \sum_{1\leq i<j\leq m} u(x_i-x_j) \,, \qquad
        \xx_m = (x_1,\dots,x_m)\,,
\]
is stable (cf.~Dobrushin~\cite{RDOB64}), i.e., for every $u\in\U$
there exists $B>0$ such that
\[
   H_u(\mathbf{x}_m) \,\geq\, -B m\,,
\]
independent of the number $m$ of particles.
The statistical distribution of the
particles of such a grand canonical ensemble is determined by the 
corresponding $m$-particle correlation functions
\be{rho-m-Lambda}
  \rho^{(m)}_{\Lambda_\ell}(\xx_m) \,=\, 
  \frac{e^{\beta\mu m}}{\Xi(\Lambda_\ell,\beta,\mu,u)} \sum_{N=0}^\infty \frac{e^{\beta\mu N}}{N!}
  \int_{\Lambda_\ell^N} e^{-\beta H_u(\xx_m,\yy_N)}\diff\yy_N\,, 
\ee
where $\xx_m=(x_1,\dots,x_m)\in\Lambda_\ell^m$,
$\yy_N=(y_1,\dots,y_N)\in\Lambda_\ell^N$, 
the integral $\int_{\Delta^0}c\diff\xx_0$ with bounded
domain $\Delta\subset\R^d$ is always taken to be equal to $c$,
and the normalizing constant 
\[
  \Xi(\Lambda_\ell,\beta,\mu,u) \,=\,
  \sum_{N=0}^\infty \frac{e^{\beta\mu N}}{N!}
  \int_{\Lambda_\ell^N} e^{-\beta H_u(\xx_N)}\diff\xx_N
\]
is the associated grand canonical partition function. In \req{rho-m-Lambda}
$m$ varies in $\N_0$, with $\rho^{(0)}_{\Lambda_\ell}$ being set to one.

We assume that for some sequence $(\ell_k)_k$ going to infinity as $k\to\infty$,
these correlation functions converge uniformly on compact
subsets to translation invariant functions $\rho^{(m)}:(\R^d)^m\to\R^+_0$, 
$m\in\N_0$, defined on the entire space;
in particular, this implies that $\rho^{(1)}$ is a constant.
It is known that these limiting correlation functions satisfy a 
so-called \emph{Ruelle bound}, i.e., 
\be{Ruelle-bound}
   \sup_{\mathbf{x}_m\in(\R^d)^m} \rho^{(m)}(\mathbf{x}_m) \,\leq\, \xi^m\,, \qquad
   m\in\N_0\,,
\ee
for some $\xi>0$, depending only on $\mu$, $\beta$, and $u$,
and that they define a translation invariant probability measure $\mathbb{P}$ 
on the configuration space
\[
   \Gamma \,=\, \bigl\{\, \gamma \subset \mathbb{R}^d \ \Big\vert\ 
                   \Delta \subset \mathbb{R}^d \text{ bounded} 
                   \,\Rightarrow \, \# (\gamma\cap \Delta)<\infty \bigr\}\,,
\]
i.e., the set of all locally finite subsets of $\mathbb{R}^d$
representing the positions of the (at most countably many) individual particles 
in space, equipped with an appropriate $\sigma$-algebra, 
cf.~Ruelle~\cite{DRUE70}. This means that
if $m\in\mathbb{N}_0$ is fixed and an  observable $F$ depends on all possible $m$-tuples of particles in a given configuration, i.e., 
\[
	F(\gamma)=\sum_{x_1,\dots,x_m\in\gamma\atop x_i\neq x_j}
 f(\mathbf{x}_m)
\]
for some $f=f_1+f_2$ with $f_1\in L^1((\R^d)^m)$ and $f_2\geq 0$, then
\be{rho-m}
  \int_\Gamma F(\gamma)\diff\mathbb{P}(\gamma)
  \,=\,\int_{(\R^d)^m} f(\mathbf{x}_m)\rho^{(m)}(\mathbf{x}_m)\diff\mathbf{x}_m
\ee
is the expected value of the corresponding observable.
In particular, if $|\Delta|$ denotes the volume of any bounded domain
$\Delta\subset\R^d$ then 
\[
   \rho(\mathbb{P}) 
   = \frac{1}{|\Delta|}\int_\Gamma \#(\gamma\cap\Delta) \diff\mathbb{P}(\gamma)
   = \rho^{(1)}
\]
is the limiting particle counting density.

%
According to \cite{DRUE70},
$\mathbb{P}$ is a translation invariant tempered 
($\beta,\mu,u$)-\emph{Gibbs measure}, 
denoted $\mathbb{P}\in\P(\beta,\mu,u)$. This means that $\mathbb{P}$
is supported by the set of tempered configurations (defined in \cite{DRUE70}),
and that
for every $F\in L^1(\mathbb{P})$ and every bounded domain $\Delta\subset\R^d$
there holds
\be{ruelle-eq}
\begin{aligned}
   &\int_\Gamma F(\gamma)\diff\mathbb{P}(\gamma) \\
   &\qquad =\,\sum_{N= 0}^\infty \frac{z^N}{N!} \int_{\Delta^N}\!
	\Biggl( \int_{\Gamma(\Delta^c)}F(\gamma')
	e^{-\beta W_u(\mathbf{x}_N;\gamma)} \diff \mathbb{P}(\gamma)\!\Biggr)
	e^{-\beta H_u(\mathbf{x}_N)}\diff \mathbf{x}_N\,,
\end{aligned}
\ee
where $\gamma'=\gamma\cup\{x_1,\dots,x_N\}$,
$\Gamma(\Delta^c)=\bigl\{\gamma\in\Gamma\,:\, 
\gamma\subset\R^d\setminus\Delta\bigr\}$,
and the interaction $W_u$ between particles at $x_i$, $i=1,\dots,N$ and those of  $\gamma\in\Gamma$ is defined as
\be{definteraction}
	W_u(\mathbf{x}_N;\gamma) = \sum_{i=1}^N\sum_{y\in\gamma} u(x_i-y)\,,
\ee
if the series converges absolutely, and as $+\infty$ otherwise. 

Given the limiting correlation functions one can define
Janossy densities $j^{(m)}_{\Lambda_\ell}:\Lambda_\ell^m\to\R$ for 
every $m\in\N_0$ and $\ell>0$ via
\be{janossyformula} 
   j_{\Lambda_\ell}^{(m)}(\mathbf{x}_m)
   \,=\, \sum_{k = 0}^\infty \frac{(-1)^k}{k!}
            \int_{\Lambda_\ell^k}\rho^{(m+k)}(\mathbf{x}_{m},\mathbf{y}_k)
            \diff\mathbf{y}_k\,.
\ee
These Janossy densities provide the induced probability density on 
$\Lambda_\ell$, for which
\be{localobservable}
   \int_\Gamma F(\gamma)\diff\mathbb{P}(\gamma)
   \,=\, \sum_{m=0}^\infty \frac{1}{m!} \int_{\Lambda_\ell^m} 
            f_m(\xx_m) j^{(m)}_{\Lambda_\ell}(\xx_m) \diff\xx_m
\ee
for every $F\in L^1(\mathbb{P})$, which satisfies
$F(\gamma)=F(\gamma\cap\Lambda_\ell)$, and which is given by functions
$f_m:\Lambda^m\to\R$, $m\in\N_0$, such that $F(\gamma_m)=f_m(\xx_m)$,
when $\gamma_m=\{x_1,\dots,x_m\}\subset\Lambda$.
Such observables $F$ are thus called \emph{local observables}.

Varying $\beta>0$, $\mu\in\R$, and $u\in\U$, we denote by 
\[
  \P=\bigcup_{\beta,\mu,u}\P(\beta,\mu,u)
\]
the union of all translation invariant tempered Gibbs measures, 
some of which may 
not be obtained as limits of finite-volume Gibbs measures with empty boundary 
conditions (cf., e.g., Georgii~\cite{HGEO94}). 
We mention for later use that for almost every $x\in\mathbb{R}^d$ and $\mathbb{P}$ almost surely for every $\mathbb{P}\in \P$
the interaction defined in \req{definteraction} is finite,
and there holds
\be{asconvergence}
	\lim_{\ell\to\infty}W_{u}(x;\gamma\cap\Lambda_\ell)= W_{u}(x;\gamma)\,,
\ee 
see \cite[Sect.~5]{TKUN03}.
We mention further that there exists some 
$\mu_0\in\R$ depending on $\beta$ and $u\in\U$, such that for 
$\mu<\mu_0$ -- the so-called \emph{gas phase} --
the set $\P(\beta,\mu,u)$ consists of a single Gibbs measure $\mathbb{P}$ only,
cf.~\cite{DRUE70},
and that in this case the correlation functions $\rho^{(m)}_{\Lambda_\ell}$
converge to $\rho^{(m)}$ as $\ell\to\infty$ for every $m\in\N_0$, 
cf.~\cite{DRUE69}.

\section{The Gibbs variational principle}
\label{Sec:Gibbs}
The Gibbs variational principle goes back to Gibbs' work
(cf.~\cite[p.~131]{GIBB02}) and appears in different variants in
statistical mechanics and stochastic analysis; we refer, e.g., to the books by
Ruelle~\cite{DRUE69}, Gallavotti~\cite{GGAL99}, and Georgii~\cite{HGEO11}
for rigorous mathematical treatments of this variational principle. 
Here we apply a particular version
established by Georgii and Zessin in the series of papers 
\cite{HGEO93,HGEO94,HGEO95}. 

For pair potentials $u\in\U$ and Gibbs measures $\mathbb{P}\in\P$,
we introduce the \emph{specific energy}
\be{energy}
   E(u,\mathbb{P}) = 
   \lim_{\ell \to \infty}\frac{1}{|\Lambda_\ell|}
   \int_{\Gamma} 
      \frac{1}{2}\!\!\!\sum_{x,y\in \gamma\cap\Lambda_\ell \atop x\neq y}\!\!\! u(x-y)
   \diff \mathbb{P}(\gamma)
\ee
and the \emph{specific (relative) entropy}
\[
   S(\mathbb{P}) = \lim_{\ell \to \infty } \frac{1}{|\Lambda_\ell|}
                    \sum_{m=0}^\infty\frac{1}{m!}\int_{\Lambda_\ell^m}
                       j_{\Lambda_\ell}^{(m)}(\mathbf{x}_m)
                       \log\bigl(j_{\Lambda_\ell}^{(m)}(\mathbf{x}_m)\bigr)
                    \diff \mathbf{x}_m\,,
\]
where both limits are known to exist in $\R\cup\{+\infty\}$:
concerning the specific energy
see Proposition~\ref{lem:energyeasy} below; concerning the specific entropy
we refer to Robinson and Ruelle~\cite{DRUE67} --
in fact, using \req{janossyformula} and \req{Ruelle-bound}
it is not too difficult to see that $S(\mathbb{P})$ is finite
for every $\mathbb{P}\in\P$. 
The relative entropy differs from the standard entropy by a sign, 
to simplify language we call $S(P)$ nevertheless the entropy. 
We take similar liberties with the sign of the 
(specific) \emph{grand potential}
\begin{align}\label{eq:freeenergyfunctional}
   \gp_{\beta,\mu}(u,\mathbb{P})
   = \mu\thsp\rho(\mathbb{P}) 
     - E(u,\mathbb{P}) - \frac{1}{\beta}\thsp S(\mathbb{P})\,,
\end{align}
for which the following variational principle holds true
(\cite[Theorem~3.4]{HGEO95}).

\begin{theoremold}
\label{thm:gvp}
For fixed $\mu\in\R$, $\beta>0$, and $u\in\U$ 
the grand potential $\gp_{\beta,\mu}(u,\,\cdot\,)$ has values in 
$\R\cup\{-\infty\}$. Its maximal value $p$ on $\P$ is attained
for every $\mathbb{P}\in\P(\beta,\mu,u)$,
and there holds
\[
  \gp_{\beta,\mu}(u,\mathbb{P}) < p
\]
for every other $\mathbb{P}\in\P$. Here, 
\[
   p \,=\, \lim_{\ell\to\infty} \frac{1}{\beta|\Lambda_\ell|}
           \log\Xi(\Lambda_\ell,\beta,\mu,u)
\]
is the \emph{pressure} in the thermodynamical limit.
\end{theoremold}

For the proof of the Henderson result we will also need the following identity,
for which we include a self-contained proof for the ease of the reader.

\begin{proposition}\label{lem:energyeasy}
For every $u\in\mathscr{U}$ and every $\mathbb{P}\in\P$
the limit in \req{energy} belongs to $\R\cup\{+\infty\}$, and is given by
\be{energyeasy}
   E(u,\mathbb{P}) = \frac{1}{2}\int_{\mathbb{R}^d} u(x)\rho^{(2)}(x,0)\diff x\,,
\ee
where $\rho^{(2)}$ is the pair correlation function associated with 
$\mathbb{P}$.
\end{proposition}

\begin{proof}
Let $\rho^{(2)}$ be the translation invariant pair correlation function 
associated with $\mathbb{P}$. Then we can apply \req{rho-m} with $m=2$ and
\[
   f(x,y) = \begin{cases}
               u(x-y)\,, & x,y\in\Lambda_\ell\,,\\
               0\,,        & \text{else}\,,
            \end{cases}
\]
to rewrite
\begin{align}
\nonumber
  &\int_\Gamma \sum_{x,y\in\gamma\cap\Lambda_\ell\atop x\neq y}
     u(x-y)\,\diff\mathbb{P}(\gamma)
   \,=\,\int_{\Lambda_\ell^2} u(x-y)\rho^{(2)}(x,y)\diff(x,y) \\
\label{eq:deltasplit}
  &=\, \int_{\Lambda_\ell}\!\!\left(\int_{\Delta_{x,\ell}}u(x-y)\rho^{(2)}(x,y)\diff y
       +\int_{\Lambda_\ell\backslash\Delta_{x,\ell}} u(x-y)\rho^{(2)}(x,y)\diff y\right)
       \! \diff x,
\end{align}
where the set 
$\Delta_{x,\ell} = \left\lbrace y\in\Lambda_\ell \mid u(x-y) \geq 1 
\right\rbrace$ is bounded, and $u(x-\,\cdot\,)$ is absolutely integrable 
over $\R^d\setminus\Delta_{x,\ell}$ because of \req{majorant}.
Therefore, and since $\rho^{(2)}$ is bounded, compare~\req{Ruelle-bound},
the second inner integral in \req{deltasplit} is uniformly bounded,
independent of $\ell$ and $x\in\Lambda_\ell$.
The integrand of the first inner integral is nonnegative. 
In case this integral diverges for some $\ell\in\N$ and some $x\in\Lambda_\ell$
then the total right-hand side of \req{deltasplit} equals $+\infty$, and
this remains true for all larger values of $\ell$, 
i.e., $E(u,\mathbb{P})=+\infty$. The same argument applied to the
right-hand side of \req{energyeasy} shows that equality holds in
\req{energyeasy} in this case, because $\rho^{(2)}$ is translation invariant. 

On the other hand, if the inner integral over $\Delta_{x,\ell}$
in \req{deltasplit} is finite for every $\ell\in\N$ and every 
$x\in\Lambda_\ell$, then the same argument
as before, together with the translation invariance of $\rho^{(2)}$
shows that
\be{translinv}
   \int_{\R^d} u(x-y)\rho^{(2)}(x,y)\dy 
   \,=\, \int_{\R^d} u(y)\rho^{(2)}(y,0)\dy
\ee
is absolutely convergent.
Now we assume that $\ell$ is greater than the parameter $r_0$
which occurs in \req{majorant}.
Then we choose some $r$ between $r_0$ and $\ell$, 
and we split the domain $\Lambda_\ell^2$ of integration into
\[
   \Lambda_\ell^2 \,=\, \Delta_1 \cup \Delta_2 \cup \Delta_3\,,
\]
where
\begin{align*}
   \Delta_1 &\,=\, \{\,(x,y)\in\Lambda_\ell^2\,:\, 
                       x\in\Lambda_{\ell-r}\,, \
                       |y-x|\leq r\,\}\,,\\
   \Delta_2 &\,=\, \{\,(x,y)\in\Lambda_\ell^2\,:\,
                       x\in\Lambda_\ell\setminus\Lambda_{\ell-r}\,, \
                       |y-x|\leq r\,\}\,,\\
   \Delta_3 &\,=\, \{\,(x,y)\in\Lambda_\ell^2\,:\, |x-y|>r\,\}\,.
\end{align*}
Under these assumptions on $r$ and $\ell$ it follows from the fact that 
\req{translinv} is absolutely convergent that
\begin{align*}
   &\Bigl|\int_{\Delta_2} u(x-y)\rho^{(2)}(x,y)\diff(x,y)\Bigr|
    \,\leq\, \int_{\Lambda_\ell\setminus\Lambda_{\ell-r}} 
                \int_{|y-x|\leq r} \bigl|u(x-y)\bigr|\rho^{(2)}(x,y)\dy\dx\\[1ex]
   &\qquad \,\leq\, \bigl(|\Lambda_\ell|-|\Lambda_{\ell-r}|\bigr)
                    \int_{\R^3} \bigl|u(y)\bigr|\rho^{(2)}(y,0)\dy
\intertext{and}
   &\Bigl|\int_{\Delta_3} u(x-y)\rho^{(2)}(x,y)\diff(x,y)\Bigr|
    \,\leq\, \int_{\Delta_3} \bigl|u(x-y)\bigr|\rho^{(2)}(x,y)\diff(x,y) 
             \\[1ex]
   &\qquad \,\leq\, \int_{\Lambda_\ell}\int_{|y-x|>r} 
                       \bigl|u(x-y)\bigr|\rho^{(2)}(x,y)\dy\dx
           \,\leq\, |\Lambda_\ell|
                    \int_{|y|>r} \bigl|u(y)\bigr|\rho^{(2)}(y,0) \dy\,.
\end{align*}
Since \req{translinv} converges absolutely
we can thus choose $r=r(\varepsilon)$ sufficiently large to make sure that
\be{Delta23}
   \limsup_{\ell\to\infty}\,\Biggl|
      \frac{1}{|\Lambda_\ell|} 
      \int_{\Delta_2\cup\Delta_3}
         u(x-y)\rho^{(2)}(x,y)\diff(x,y)\,
      \Biggr| \,\leq\, \varepsilon
\ee
for any given positive number $\varepsilon$.
On the other hand, using the translation invariance again, we have
\begin{align*}
   &\int_{\Delta_1} u(x-y)\rho^{(2)}(x,y)\diff(x,y)
   \,=\, \int_{\Lambda_{\ell-r}} \int_{|y-x|\leq r}
            u(x-y)\rho^{(2)}(x,y)\dy\dx\\[1ex]
   &\qquad 
    \,=\, |\Lambda_{\ell-r}| \int_{|y|\leq r} u(y)\rho^{(2)}(y,0)\dy\,,
\end{align*}
and hence,
\be{Delta1}
\begin{aligned}
   &\limsup_{\ell\to\infty}\,\Biggl|\frac{1}{|\Lambda_\ell|}
    \int_{\Delta_1} u(x-y)\rho^{(2)}(x,y)\diff(x,y)
    \,-\, \int_{\R^d} u(y)\rho^{(2)}(y,0)\dy\,\Biggr|\\[1ex]
   &\qquad \,\leq\, \int_{|y|>r} \bigl|u(y)\bigr|\rho^{(2)}(y,0)\dy\,.
\end{aligned}
\ee
Combining \req{Delta1} for $r=r(\varepsilon)$ with \req{Delta23},
the assertion \req{energyeasy} follows by letting $\varepsilon\to 0$.
\end{proof}

We mention that the Kirkwood-Salsburg equations (cf., e.g., \cite{DRUE70})
can be used to argue that the integrand of \req{energyeasy} is bounded
near the origin when $\mathbb{P}$ is a $(\beta,\mu,u)$-Gibbs measure,
so that for ``matching'' $u$ and $\mathbb{P}$ the integral is absolutely 
convergent and finite by virtue of \req{majorant} and \req{Ruelle-bound}.

\section{Uniqueness results of Henderson type in the thermodynamical limit}
\label{Sec:Henderson}
We now formulate Henderson's theorem in the spirit of his original paper
\cite{RHEN74}, and provide a rigorous proof, based on arguments 
borrowed from \cite{RHEN74} and from the proof of \cite[Thm.~2.34]{HGEO11}.

\begin{theorem}
\label{Thm:Hen1}
Let $u,v\in\U$, $\beta>0$, and $\mu,\mu'\in\R$ be given, and assume that
$\mathbb{P}_u\in\P(\beta,\mu,u)$ and $\mathbb{P}_v\in\P(\beta,\mu',v)$
admit the same density $\rho^{(1)}$ and the same
pair correlation function $\rho^{(2)}$.
Then $\mu=\mu'$ and $u=v$ almost everywhere.
\end{theorem}

\begin{proof}
By Theorem~\ref{thm:gvp} we have
\[
\gp_{\beta,\mu}(u,\mathbb{P}_v) \leq \gp_{\beta,\mu}(u,\mathbb{P}_u) 
\quad\text{ and } \quad
\gp_{\beta,\mu'}(v,\mathbb{P}_u) \leq \gp_{\beta,\mu'}(v,\mathbb{P}_v)\,.
\]
Since $\gp_{\beta,\mu}(u,\mathbb{P}_u)$ 
and $\gp_{\beta,\mu'}(v,\mathbb{P}_v)$ 
are finite we may write these inequalities as
\begin{align}\label{eq:nonnegativevariations}
\gp_{\beta,\mu}(u,\mathbb{P}_v)
 -\gp_{\beta,\mu}(u,\mathbb{P}_u)\leq 0 
\quad \text{ and } \quad
 \gp_{\beta,\mu'}(v,\mathbb{P}_u)
 -\gp_{\beta,\mu'}(v,\mathbb{P}_v) \leq 0\,,
\end{align}
and adding them we get
\begin{align}\label{eq:nonnegativefreeenergy}
\gp_{\beta,\mu}(u,\mathbb{P}_v) -\gp_{\beta,\mu}(u,\mathbb{P}_u)
+ \gp_{\beta,\mu'}(v,\mathbb{P}_u) -
  \gp_{\beta,\mu'}(v,\mathbb{P}_v) \leq 0 \,.
\end{align}
Recalling the definition~\req{freeenergyfunctional} of the 
grand potential we have
\begin{align}
\nonumber
   &\gp_{\beta,\mu}(u,\mathbb{P}_v) -\gp_{\beta,\mu}(u,\mathbb{P}_u)
    + \gp_{\beta,\mu'}(v,\mathbb{P}_u) -\gp_{\beta,\mu'}(v,\mathbb{P}_v) \\[1ex]
\nonumber
   &\quad
    = \mu\thsp\rho (\mathbb{P}_v) 
      - E(u,\mathbb{P}_v)-\frac{1}{\beta} S(\mathbb{P}_v)
      - \mu\thsp\rho (\mathbb{P}_u) +  E(u,\mathbb{P}_u)
      + \frac{1}{\beta}\thsp S(\mathbb{P}_u)\\
\nonumber
   &\quad \qquad
      + \mu'\rho (\mathbb{P}_u) - E(v,\mathbb{P}_u) 
      - \frac{1}{\beta}\thsp S(\mathbb{P}_u)
      - \mu'\rho (\mathbb{P}_v) + E(v,\mathbb{P}_v)
      + \frac{1}{\beta}\thsp S(\mathbb{P}_v)
        \\[1ex]
\label{eq:vieleHamiltonians}
   &\quad
    =  -E(u,\mathbb{P}_v) +  E(u,\mathbb{P}_u)
       -  E(v,\mathbb{P}_u)+  E(v,\mathbb{P}_v)\,,
\end{align}
because $\rho(\mathbb{P}_u)=\rho(\mathbb{P}_v)=\rho^{(1)}$ by assumption.
Furthermore, by virtue of Proposition~\ref{lem:energyeasy} and the fact
that the pair correlation functions of $\mathbb{P}_u$ and
$\mathbb{P}_v$ coincide, there holds
\begin{align*}
   E(u,\mathbb{P}_u) = \frac{1}{2}\int_{\mathbb{R}^d} u(x)\rho^{(2)}(x,0)\diff x
   = E(u,\mathbb{P}_v)
\intertext{and}
   E(v,\mathbb{P}_u) = \frac{1}{2}\int_{\mathbb{R}^d} v(x)\rho^{(2)}(x,0)\diff x
   = E(v,\mathbb{P}_v)\,.
\end{align*}
Inserting this into \cref{eq:vieleHamiltonians} we conclude that
\begin{align*}
	\gp_{\beta,\mu}(u,\mathbb{P}_v) -\gp_{\beta,\mu}(u,\mathbb{P}_u)
	+ \gp_{\beta,\mu'}(v,\mathbb{P}_u) -\gp_{\beta,\mu'}(v,\mathbb{P}_v)=0.
\end{align*}
Accordingly, equality holds in \req{nonnegativefreeenergy}, and thus
necessarily in
both statements of \cref{eq:nonnegativevariations}. By the Gibbs variational
principle (Theorem~\ref{thm:gvp})
this implies that $\mathbb{P}_v\in\P(\beta,\mu,u)$ and 
$\mathbb{P}_u\in\P(\beta,\mu',v)$.

It therefore follows from \req{ruelle-eq} that
\[
\begin{aligned}
   \int_\Gamma F(\gamma)\diff\mathbb{P}_u(\gamma) 
   &\,=\,\sum_{N= 0}^\infty \frac{1}{N!} \int_{\Delta^N}\!
     \Biggl( \int_{\Gamma(\Delta^c)}\!\!F(\gamma')
     e^{-\beta W_u(\mathbf{x}_N;\gamma)} \diff \mathbb{P}_u(\gamma)\!\Biggr)
     e^{N\beta\mu-\beta H_u(\mathbf{x}_N)}\diff \mathbf{x}_N \\
   &\,=\,\sum_{N= 0}^\infty \frac{1}{N!} \int_{\Delta^N}\!
     \Biggl( \int_{\Gamma(\Delta^c)}\!\!F(\gamma')
     e^{-\beta W_v(\mathbf{x}_N;\gamma)} \diff \mathbb{P}_u(\gamma)\!\Biggr)
     e^{N\beta\mu'-\beta H_v(\mathbf{x}_N)}\diff \mathbf{x}_N
\end{aligned}
\]
for every $F\in L^1(\mathbb{P}_u)$ and every bounded domain $\Delta\subset\R^d$.
Therefore
\be{HuWuHvWv}
 H_u(\mathbf{x}_N)+W_u(\mathbf{x}_N;\gamma)-N\mu 
 \,=\, H_v(\mathbf{x}_N)+W_v(\mathbf{x}_N;\gamma)-N\mu'
\ee
for every $N\in\mathbb{N}$,
almost every $\mathbf{x}_N\in \Delta^N$ and  
$\mathbb{P}_u$ almost surely for $\gamma\in\Gamma(\Delta^c)$.
For $N=1$ this means that
\be{interactionequality}
	W_u(x;\gamma) - \mu \,=\, W_v(x;\gamma) - \mu'
\ee
for almost every $x\in \Delta$ and 
$\mathbb{P}_u$ and $\mathbb{P}_v$ almost surely for $\gamma\in\Gamma(\Delta^c)$.
Using the additivity of $W_u$ and $W_v$ 
in the first argument we thus conclude from the case $N=2$ of \req{HuWuHvWv}
that
\be{equality-vu}
	u(x-y)=v(x-y)
\ee
for almost every $x,y\in\Delta$, and since $\Delta$ was arbitrarily chosen, 
we have $u=v$ almost everywhere.

Inserting \req{equality-vu} into \req{definteraction} it follows that
\begin{equation*}
	W_u(x;\gamma_0) = W_v(x;\gamma_0)
\end{equation*}
for  almost every $x\in\Delta$ and almost every finite subset 
$\gamma_0\subset \Delta^c\cap\Lambda_\ell$.
Together with \req{localobservable} this implies that for every 
$\mathbb{P}\in\P$ there holds
\be{endlgleich}
	W_u(x;\gamma\cap\Lambda_\ell) = W_v(x;\gamma\cap\Lambda_\ell)
\ee
$\mathbb{P}$ almost surely for $\gamma\in\Gamma(\Delta^c)$.
By virtue of \req{asconvergence} and \req{endlgleich} we therefore have
\[
   W_u(x;\gamma) = W_v(x;\gamma)
\]
for almost every $x\in\Delta$ and $\mathbb{P}_u$ almost surely for 
$\gamma\in\Gamma(\Delta^c)$, and hence we conclude from 
\req{interactionequality}
that $\mu=\mu'$, which remained to be shown.
\end{proof}

Henderson stipulated the assumptions of Theorem~\ref{Thm:Hen1} 
from the canonical ensemble point of view. Concerning the
grand canonical perspective an analogous uniqueness result is as follows.

\begin{theorem}
\label{Thm:Hen2}
Let $u,v\in\U$, $\beta>0$, and $\mu\in\R$ be given, 
and assume that $\mathbb{P}_u\in\P(\beta,\mu,u)$ and 
$\mathbb{P}_v\in\P(\beta,\mu,v)$ 
admit the same pair correlation function $\rho^{(2)}$. 
Then $u=v$ almost everywhere.
\end{theorem}

The proof is the same as for Theorem~\ref{Thm:Hen1}: 
This time \req{vieleHamiltonians} holds true because the chemical potentials 
are the same. We mention, however, that we do not know whether the 
counting densities
of the two Gibbs measures are necessarily the same, unless it is assumed that
the corresponding $(\beta,\mu,u)$-Gibbs measure is uniquely determined --
as it is, e.g., in the gas phase.

\section{Concluding Remarks}\label{Sec:problems}
We emphasize that for our results we do not stipulate that the system
is in the gas phase, nor that the set $\P(\beta,\mu,u)$ consists of 
a single Gibbs measure only.

In case it is known that $u$ is also radially symmetric, i.e., 
if the interaction of two particles only depends on their distance, then
one can show -- using the Markov-Kakutani fixed point
theorem as in the proof of \cite[Theorem~5.8]{DRUE69}, compare
Kuna~\cite{TKUN99} -- 
that there exists at least one rotation and
translation invariant Gibbs measure $\mathbb{P}_u\in\P(\beta,\mu,u)$,
which can be used to define a 
\emph{radial distribution function} $g:\R^+\to\R_0^+$ given by
\be{g}
   g(r) \,=\, \frac{\rho^{(2)}(x_1,x_2)}{(\rho^{(1)})^2}\,, \qquad 
   r = |x_1-x_2|\,,
\ee
provided that the density is nonzero.
Assuming further that $\mathbb{P}_v\in\P(\beta,\mu,v)$
is also rotation invariant, then
one obviously can impose in Theorem~\ref{Thm:Hen1}
-- as did Henderson --
that the radial distribution functions and the densities are the same for
these two Gibbs measures, and the statement of the theorem remains valid. 
We do not know whether in the formulation of Theorem~\ref{Thm:Hen2} 
$\rho^{(2)}$ can also be replaced by the radial distribution function
in the isotropic case.

Finally we remark that the representation~\req{energyeasy} of the
specific energy is not essential for the uniqueness argument. 
By its definition~\req{energy}, 
the specific energy $E(u,\mathbb{P})$ is the limit in $\R\cup\{+\infty\}$
of the expression given in \req{deltasplit} 
normalized by two times the volume of $\Lambda_\ell$, and hence, 
its value only depends on $u$ and on the pair correlation function $\rho^{(2)}$
associated with $\mathbb{P}$.
This suffices to conclude that the expression~\req{vieleHamiltonians}
sums up to zero.

Having settled the uniqueness problem the natural follow-up question 
concerns the existence of solutions of the inverse Henderson problem,
i.e., what are necessary and
sufficient conditions on a given triplet
$\beta,\rho>0$, $\mu\in\R$, and a nonnegative translation invariant function 
$\rho^{(2)}:\R^2\to\R$, such that
there exists a pair potential $u\in\U$ for which $\rho$ is the density
and $\rho^{(2)}$ is the pair correlation function of a 
Gibbs measure $\mathbb{P}\in\P(\beta,\mu,u)$.
Partial results for this problem have been
contributed, e.g., by Caglioti, Kuna, Lebowitz, and Speer~\cite{CKLS06},
Kuna, Lebowitz and Speer~\cite{KLS07}, and Koralov~\cite{LKOR07}.
The general existence problem is widely open, though.


\bibliographystyle{siam}

\begin{thebibliography}{10}

\bibitem{CKLS06}
{\sc E.~Caglioti, T.~Kuna, J.~Lebowitz, and E.~Speer}, {\em Point processes
  with specified low order correlations}, Markov Process. Related Fields, 12
  (2006), pp.~257--272.

\bibitem{RDOB64}
{\sc R.~Dobrushin}, {\em Investigation of conditions for the asymptotic
  existence of the configuration integral of {G}ibbs’ distribution}, Theor.
  Probability Appl., 9 (1964), pp.~566--581.

\bibitem{GGAL99}
{\sc G.~Gallavotti}, {\em Statistical Mechanics: A Short Treatise}, Springer,
  Berlin, 1999.

\bibitem{HGEO94}
{\sc H.-O. Georgii}, {\em Large deviations and the equivalence of ensembles for
  {G}ibbsian particle systems with superstable interaction}, Probab. Theory
  Related Fields, 99 (1994), pp.~171--195.

\bibitem{HGEO95}
\leavevmode\vrule height 2pt depth -1.6pt width 23pt, {\em The equivalence of
  ensembles for classical systems of particles}, J. Stat. Phys., 80 (1995),
  pp.~1341--1378.

\bibitem{HGEO11}
\leavevmode\vrule height 2pt depth -1.6pt width 23pt, {\em Gibbs measures and
  phase transitions}, De Gruyter, Berlin, 2.~ed., 2011.

\bibitem{HGEO93}
{\sc H.-O. Georgii and H.~Zessin}, {\em Large deviations and the maximum
  entropy principle for marked point random fields}, Probab. Theory Related
  Fields, 96 (1993), pp.~177--204.

\bibitem{GIBB02}
{\sc J.~Gibbs}, {\em Elementary Principles in Statistical Mechanics},
  Scribner's Sons, New York, 1902.

\bibitem{GR71}
{\sc R.~Griffiths and D.~Ruelle}, {\em Strict convexity (``continuity'') of the
  pressure in lattice systems}, Comm. Math. Phys., 23 (1971), pp.~169--175.

\bibitem{RHEN74}
{\sc R.~Henderson}, {\em A uniqueness theorem for fluid pair correlation
  functions}, Phys. Lett. A, 49 (1974), pp.~197--198.

\bibitem{HK64}
{\sc P.~Hohenberg and W.~Kohn}, {\em Inhomogeneous electron gas}, Phys. Rev.,
  136 (1964), pp.~B\,864--B\,871.

\bibitem{TKUN03}
{\sc Y.~G. Kondratiev and T.~Kuna}, {\em Correlation functionals for {G}ibbs
  measures and {R}uelle bounds}, Methods Funct. Anal. Topology, 9 (2003),
  pp.~9--58.

\bibitem{LKOR07}
{\sc L.~Koralov}, {\em An inverse problem for {G}ibbs fields with hard core
  potential}, J. Math. Phys., 48 (2007).
\newblock 053301.

\bibitem{TKUN99}
{\sc T.~Kuna}, {\em Studies in configuration space analysis and applications},
  PhD thesis, Rheinische Friedrich-Wilhelms-Universit\"at, Bonn, 1999.

\bibitem{KLS07}
{\sc T.~Kuna, J.~Lebowitz, and E.~Speer}, {\em Realizability of point
  processes}, J. Stat. Phys., 129 (2007), pp.~417--439.

\bibitem{DMER65}
{\sc N.~Mermin}, {\em Thermal properties of the inhomogeneous electron gas},
  Phys. Rev., 137 (1965), pp.~A\,1471--A\,1473.

\bibitem{DRUE67}
{\sc D.~W. Robinson and D.~Ruelle}, {\em Mean entropy of states in classical
  statistical mechanics}, Comm. Math. Phys., 5 (1967), pp.~288--300.

\bibitem{DRUE69}
{\sc D.~Ruelle}, {\em Statistical Mechanics: Rigorous Results}, W.A. Benjamin
  Publ., New York, 1969.

\bibitem{DRUE70}
\leavevmode\vrule height 2pt depth -1.6pt width 23pt, {\em Superstable
  interactions in classical statistical mechanics}, Comm. Math. Phys., 18
  (1970), pp.~127--159.

\end{thebibliography}

\end{document}